\documentclass{scrartcl}
\def\arXiv{}

\usepackage{basicstuff}

\newcommand{\eps}{\varepsilon}

\newcommand{\cupdot}{\mathbin{\mathaccent\cdot\cup}}

\DeclareMathOperator{\skel}{skel}
\DeclareMathOperator{\pert}{pert}
\DeclareMathOperator{\twin}{twin}

\ifdefined\LNCS

\let\doendproof\endproof
\renewcommand\endproof{~\hfill$\qed$\doendproof}

\title{A New Perspective on Clustered Planarity \\as a Combinatorial
  Embedding Problem\thanks{Partly done within GRADR -- EUROGIGA
    project no. 10-EuroGIGA-OP-003.  Supported by a fellowship within
    the Postdoc-Program of the German Academic Exchange Service
    (DAAD).}}

\author{Thomas Bläsius \and Ignaz Rutter}

\institute{Faculty of Informatics, Karlsruhe Institute of Technology
  (KIT), Germany\\%
  \email{\{blaesius,rutter\}@kit.edu}}

\fi

\ifdefined\arXiv

\usepackage{arXiv}

\title{\Large A New Perspective on Clustered Planarity \\as a
  Combinatorial Embedding Problem\thanks{Partly done within GRADR --
    EUROGIGA project no. 10-EuroGIGA-OP-003.  Supported by a
    fellowship within the Postdoc-Program of the German Academic
    Exchange Service (DAAD).}}

\author{Thomas Bläsius \myand Ignaz Rutter}

\date{Karlsruhe Institute of Technology (KIT), Germany\\%
  \medskip \texttt{\{blaesius,rutter\}@kit.edu}}

\fi

\begin{document}
\maketitle

\begin{abstract}
  The clustered planarity problem (c-planarity) asks whether a
  hierarchically clustered graph admits a planar drawing such that the
  clusters can be nicely represented by regions.  We introduce the
  cd-tree data structure and give a new characterization of
  c-planarity.  It leads to efficient algorithms for c-planarity
  testing in the following cases.
  \begin{inparaenum}[(i)]
  \item Every cluster and every co-cluster (complement of a cluster)
    has at most two connected components.
  \item Every cluster has at most five outgoing edges.
  \end{inparaenum}

  Moreover, the cd-tree reveals interesting connections between
  c-planarity and planarity with constraints on the order of edges
  around vertices.  On one hand, this gives rise to a bunch of new
  open problems related to c-planarity, on the other hand it provides
  a new perspective on previous results.
\end{abstract}


\section{Introduction}
\label{sec:introduction}

When visualizing graphs whose nodes are structured in a hierarchy, one
usually has two objectives.  First, the graph should be drawn nicely.
Second, the hierarchical structure should be expressed by the drawing.
Regarding the first objective, we require drawings without edge
crossings, i.e., \emph{planar drawings} (the number of crossings in a
drawing of a graph is a major aesthetic criterion).  A natural way to
represent a cluster is a simple region containing exactly the vertices
in the cluster.  To express the hierarchical structure, the boundaries
of two regions must not cross and edges of the graph can cross region
boundaries at most once, namely if only one of its endpoints lies
inside the cluster.  Such a drawing is called \emph{c-planar}; see
Section~\ref{sec:preliminaries} for a formal definition.  Testing a
clustered graph for \emph{c-planarity} (i.e., testing whether it
admits a c-planar drawing) is a fundamental open problem in the field
of Graph Drawing.

C-planarity was first considered by Lengauer~\cite{l-hpta-89} but in a
completely different context.  He gave an efficient algorithm for the
case that every cluster is connected.  Feng et al.~\cite{fce-pcg-95},
who coined the name c-planarity, rediscovered the problem and gave a
similar algorithm.
Cornelsen and Wagner~\cite{cw-cccg-06} showed that c-planarity is
equivalent to planarity when every cluster and every co-cluster is
connected.

Relaxing the condition that every cluster must be connected, makes
testing c-planarity surprisingly difficult.  Efficient algorithms are
known only for very restricted cases and many of these algorithms are
very involved.  One example is the efficient algorithm by Jel\'inek et
al.~\cite{jjkl-cp-09,jjkl-cp-09-long} for the case that every cluster
consists of at most two connected components while the planar
embedding of the underlying graph is fixed.  Another efficient
algorithm by Jel\'inek et al.~\cite{jstv-cp-09} solves the case that
every cluster has at most four outgoing edges.

A popular restriction is to require a \emph{flat} hierarchy, i.e.,
every pair of clusters has empty intersection.  For example, Di
Battista and Frati~\cite{g-ecpte-08} solve the case where the
clustering is flat, the graph has a fixed planar embedding and the
size of the faces is bounded by five.
Section~\ref{sec:related-work-flat} and
Section~\ref{sec:related-work-connected} contain additional related
work viewed from the new perspective.

\subsection{Contribution \& Outline}
\label{sec:contr-outl}

We first present the cd-tree data structure
(Section~\ref{sec:cd-tree}), which is similar to a data structure used
by Lengauer~\cite{l-hpta-89}.  We use the cd-tree to characterize
c-planarity in terms of a combinatorial embedding problem.  We believe
that our definition of the cd-tree together with this characterization
provides a very useful perspective on the c-planarity problem and
significantly simplifies some previous results.

In Section~\ref{sec:clust-constr-plan} we define different
constrained-planarity problems.  We use the cd-tree to show in
Section~\ref{sec:flat-clustered-graph} that these problems are
equivalent to different versions of the c-planarity problem on
flat-clustered graphs.  We also discuss which cases of the constrained
embedding problems are solved by previous results on c-planarity of
flat-clustered graphs.  Based on these insights, we derive a generic
algorithm for testing c-planarity with different restrictions in
Section~\ref{sec:gener-clust-graphs} and discuss previous work in this
context.

In Section~\ref{sec:c-planarity-pq}, we show how the cd-tree
characterization together with results on the problem
\textsc{Simultaneous PQ-Ordering}~\cite{br-spoacep-13} lead to
efficient algorithms for the cases that
\begin{inparaenum}[(i)]
\item every cluster and every co-cluster consists of at most two
  connected components; or
\item every cluster has at most five outgoing edges.
\end{inparaenum}
The latter extends the result by Jel\'inek et al.~\cite{jstv-cp-09},
where every cluster has at most four outgoing edges.

\section{Preliminaries}
\label{sec:preliminaries}

We denote graphs by $G$ with vertex set $V$ and edge set $E$.  We
implicitly assume graphs to be \emph{simple}, i.e., they do not have
multiple edges or loops.  Sometimes we allow multiple edges (we never
allow loops).  We indicate this with the prefix \emph{multi-}, e.g., a
multi-cycle is a graph obtained from a cycle by multiplying edges.

A (multi-)graph $G$ is \emph{planar} if it admits a planar drawing (no
edge crossings).  The \emph{edge-ordering} of a vertex $v$ is the
clockwise cyclic order of its incident edges in a planar drawing of
$G$.  A \emph{(planar) embedding} of $G$ consists of an edge-ordering
for every vertex such that $G$ admits a planar drawing with these
edge-orderings. 

A \emph{PQ-tree}~\cite{bl-tcopi-76} is a tree $T$ (in our case
unrooted) with leaves $L$ such that every inner node is either a
\emph{P-node} or a \emph{Q-node}.  When embedding $T$, one can choose
the (cyclic) edge-orderings of P-nodes arbitrarily, whereas the
edge-orderings of Q-nodes are fixed up to reversal.  Every such
embedding of $T$ defines a cyclic order on the leaves $L$.  The
PQ-tree $T$ \emph{represents} the orders one can obtain in this way.
A set of orders is \emph{PQ-representable} if it can be represented by
a PQ-tree.  It is not hard to see that the valid edge-orderings of
non-cutvertices in planar graphs are PQ-representable
(e.g.,~\cite{br-spoacep-13}).  Conversely, adding wheels around the
Q-nodes of a PQ-tree $T$ and connecting all leaves with a vertex $v$
yields a planar graph $G$ where the edge-orderings of $v$ in
embeddings of $G$ are represented by $T$ (e.g.,~\cite{l-hpta-89}).

\subsection{C-Planarity on the Plane and on the Sphere}
\label{sec:clustered-plan-cut-plan}

A \emph{clustered graph} $(G, T)$ is a graph $G$ together with a
rooted tree $T$ whose leaves are the vertices of $G$.  Let $\mu$ be a
node of $T$.  The tree $T_\mu$ is the subtree of $T$ consisting of all
successors of $\mu$ together with the root $\mu$.  The graph induced
by the leaves of $T_\mu$ is a \emph{cluster} in $G$.  We identify this
cluster with the node $\mu$.  We call a cluster \emph{proper} if it is
neither the whole graph (\emph{root cluster}) nor a single vertex (\emph{leaf
cluster}).

A \emph{c-planar drawing} of $(G, T)$ is a planar drawing of $G$ in
the plane together with a \emph{simple} (=~simply-connected) region
$R_\mu$ for every cluster $\mu$ satisfying the following properties.
\begin{inparaenum}[(i)]
\item Every region $R_\mu$ contains exactly the vertices of the
  cluster $\mu$.
\item Two regions have non-empty intersection only if one contains the
  other.
\item Edges cross the boundary of a region at most once.
\end{inparaenum}
A clustered graph is \emph{c-planar} if and only if it admits a
c-planar drawing.

The above definition of c-planarity relies on embeddings in the plane
using terms like ``outside'' and ``inside''.  Instead, one can
consider drawings on the sphere by considering the tree $T$ to be
unrooted instead of rooted, using cuts instead of clusters, and simple
closed curves instead of simple regions.  Let $\eps$ be an edge in
$T$.  Removing $\eps$ splits $T$ in two connected components.  As the
leaves of $T$ are the vertices of $G$, this induces a
\emph{corresponding cut} $(V_\eps, V_\eps')$ with $V_\eps' =
V\setminus V_\eps$ on~$G$.  For a c-planar drawing of $G$ on the
sphere, we require a planar drawing of $G$ together with a simple
closed curve $C_\eps$ for every cut $(V_\eps, V_\eps')$ with the
following properties.
\begin{inparaenum}[(i)]
\item The curve $C_\eps$ separates $V_\eps$ from $V_\eps'$.
\item No two curves intersect.
\item Edges of $G$ cross $C_\eps$ at most once.
\end{inparaenum}

Note that using clusters instead of cuts corresponds to orienting the
cuts, using one side as the cluster and the other side as the
cluster's complement (the \emph{co-cluster}).  This notion of
c-planarity on the sphere is equivalent to the one on the plane; one
simply has to choose an appropriate point on the sphere to lie in the
outer face.  The unrooted view has the advantage that it is more
symmetric (i.e., there is no difference between clusters and
co-clusters), which is sometimes desirable.  We use the rooted and
unrooted view interchangeably.

\section{The CD-Tree}
\label{sec:cd-tree}

Let $(G,T)$ be a clustered graph.  We introduce the \emph{cd-tree
  (cut- or cluster-decomposition-tree)} by enhancing each node of $T$
with a multi-graph that represents the decomposition of $G$ along its
cuts corresponding to edges in $T$; see
Fig.~\ref{fig:cluster-decomposition-tree}a and~b for an example.  We
note that Lengauer~\cite{l-hpta-89} uses a similar structure.  Our
notation is inspired by SPQR-trees~\cite{dt-omtc-96}.

Let $\mu$ be a node of $T$ with neighbors $\mu_1, \dots, \mu_k$ and
incident edges $\eps_i = \{\mu, \mu_i\}$ (for $i = 1, \dots, k$).
Removing $\mu$ separates $T$ into $k$ subtrees $T_1, \dots, T_k$.  Let
$V_1, \dots, V_k \subseteq V$ be the vertices of $G$ represented by
leaves in these subtrees.  The \emph{skeleton} $\skel(\mu)$ of $\mu$
is the multi-graph obtained from $G$ by contracting each subset $V_i$
into a single vertex $\nu_i$ (the resulting graph has multiple edges
but we remove loops).  These vertices $\nu_i$ are called \emph{virtual
  vertices}.
Note that skeletons of inner nodes of $T$ contain only virtual
vertices, while skeletons of leaves consist of one virtual and one
non-virtual vertex.
The node $\mu_i$ is the neighbor of $\mu$ \emph{corresponding} to
$\nu_i$ and the virtual vertex in $\skel(\mu_i)$ corresponding to
$\mu$ is the \emph{twin} of $\nu_i$, denoted by $\twin(\nu_i)$.  Note
that $\twin(\twin(\nu_i)) = \nu_i$.

The edges incident to $\nu_i$ are exactly the edges of $G$ crossing
the cut corresponding to the tree edge $\eps_i$.  Thus, the same edges
of $G$ are incident to $\nu_i$ and $\twin(\nu_i)$.  This gives a bound
on the total size $c$ of the cd-tree's skeletons (which we shortly
call the \emph{size of the cd-tree}).  The total number of edges in
skeletons of $T$ is twice the total size of all cuts represented by
$T$.  As edges might cross a linear number of cuts (but obviously not
more), the cd-tree has at most quadratic size in the number of
vertices of $G$, i.e., $c \in O(n^2)$.

Assume the cd-tree is rooted.  Recall that in this case every node
$\mu$ represents a cluster of $G$.  In analogy to the notion for
SPQR-trees, we define the \emph{pertinent graph} $\pert(\mu)$ of the
node $\mu$ to be the cluster represented by $\mu$.  Note that one
could also define the pertinent graph recursively, by removing the
virtual vertex corresponding to the parent of $\mu$ (the \emph{parent
  vertex}) from $\skel(\mu)$ and replacing each remaining virtual
vertex by the pertinent graph of the corresponding child of~sy$\mu$.
Clearly, the pertinent graph of a leaf of $T$ is a single vertex and
the pertinent graph of the root is the whole graph $G$.  A similar
concept, also defined for unrooted cd-trees, is the \emph{expansion
  graph}.  The expansion graph $\exp(\nu_i)$ of a virtual vertex
$\nu_i$ in $\skel(\mu)$ is the pertinent graph of its corresponding
neighbor $\mu_i$ of $\mu$, when rooting $T$ at $\mu$.  One can think
of the expansion graph $\exp(\nu_i)$ as the subgraph of $G$
represented by $\nu_i$ in $\skel(\mu)$.  As mentioned before, we use
the rooted and unrooted points of view interchangeably.

The leaves of a cd-tree represent singleton clusters that exist only
due to technical reasons.  It is often more convenient to consider
cd-trees with all leaves removed as follows.  Let $\mu$ be a node with
virtual vertex $\nu$ in $\skel(\mu)$ that corresponds to a leaf.  The
leaf contains $\twin(\nu)$ and a non-virtual vertex $v \in V$ in its
skeleton (with an edge between $\twin(\nu)$ and $v$ for each edge
incident to $v$ in $G$).  We replace $\nu$ in $\skel(\mu)$ with the
non-virtual vertex $v$ and remove the leaf containing $v$.  Clearly,
this preserves all clusters except for the singleton cluster.
Moreover, the graph $G$ represented by the cd-tree remains unchanged
as we replaced the virtual vertex $\nu$ by its expansion graph
$\exp(\nu) = v$.  In the following we always assume the leaves of
cd-trees to be removed.

\subsection{The CD-Tree Characterization}
\label{sec:cd-tree-char}

We show that c-planarity testing can be expressed in terms of
edge-orderings in embeddings of the skeletons of $T$.  

\begin{figure}[tb]
  \centering
  \includegraphics[page=1]{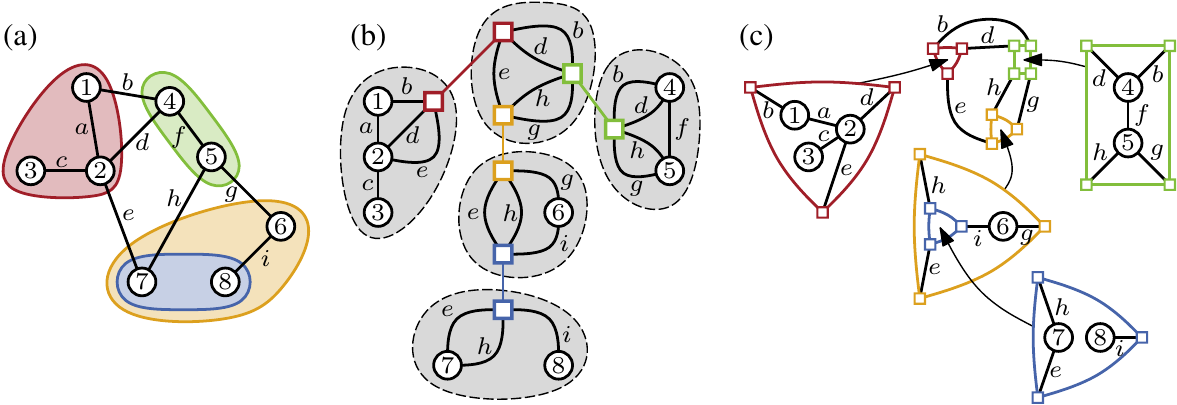}
  \caption{(a)~A c-planar drawing of a clustered graph.  (b)~The
    corresponding (rooted) cd-tree (without leaves).  The skeletons
    are drawn inside their corresponding (gray) nodes.  Every pair of
    twins (boxes with the same color) has the same edge-ordering.
    (c)~Construction of a c-planar drawing from the cd-tree.}
  \label{fig:cluster-decomposition-tree}
\end{figure}

\begin{theorem}
  \label{thm:characterization}
  A clustered graph is c-planar if and only if the skeletons of all
  nodes in its cd-tree can be embedded such that every virtual vertex
  and its twin have the same edge-ordering.
\end{theorem}
\begin{proof}
  Assume $G$ admits a c-planar drawing $\Gamma$ on the sphere.  Let
  $\mu$ be a node of $T$ with incident edges $\eps_1, \dots, \eps_k$
  connecting $\mu$ to its neighbors $\mu_1, \dots, \mu_k$,
  respectively.  Let further $\nu_i$ be the virtual vertex in
  $\skel(\mu)$ corresponding to $\mu_i$ and let $V_i$ be the nodes in
  the expansion graph $\exp(\nu_i)$.  For every cut $(V_i, V_i')$
  (with $V_i' = V \setminus V_i$), $\Gamma$ contains a simple closed
  curve $C_i$ representing it.  Since the $V_i$ are disjoint, we can
  choose a point on the sphere to be the outside such that $V_i$ lies
  inside $C_i$ for $i = 1, \dots, k$.  Since $\Gamma$ is a c-planar
  drawing, the $C_i$ do not intersect and only the edges of $G$
  crossing the cut $(V_i, V_i')$ cross $C_i$ exactly once.  Thus, one
  can contract the inside of $C_i$ to a single point while preserving
  the embedding of $G$.  Doing this for each of the curves $C_i$
  yields the skeleton $\skel(\mu)$ together with a planar embedding.
  Moreover, the edge-ordering of the vertex $\nu_i$ is the same as the
  order in which the edges cross the curve $C_i$, when traversing
  $C_i$ in clockwise direction.  Applying the same construction for
  the neighbour $\mu_i$ corresponding to $\nu_i$ yields a planar
  embedding of $\skel(\mu_i)$ in which the edge-ordering of
  $\twin(\nu_i)$ is the same as the order in which these edges cross
  the curve $C_i$, when traversing $C_i$ in counter-clockwise
  direction.  Thus, in the resulting embeddings of the skeletons, the
  edge-ordering of a virtual vertex and its twin is the same up to
  reversal.  To make them the same one can choose a 2-coloring of $T$
  and mirror the embeddings of all skeletons of nodes in one color
  class.

  Conversely, assume that all skeletons are embedded such that every
  virtual vertex and its twin have the same edge-ordering.  Let $\mu$
  be a node of $T$.  Consider a virtual vertex $\nu_i$ of $\skel(\mu)$
  with edge-ordering $e_1, \dots, e_\ell$.  We replace $\nu_i$ by a
  cycle $C_i = (\nu_i^1, \dots, \nu_i^\ell)$ and attach the edge $e_j$
  to the vertex $\nu_i^j$; see
  Fig.~\ref{fig:cluster-decomposition-tree}c.  Recall that
  $\twin(\nu_i)$ has in $\skel(\mu_i)$ the same incident edges $e_1,
  \dots, e_\ell$ and they also appear in this order around
  $\twin(\nu_i)$.  We also replace $\twin(\nu_i)$ by a cycle of length
  $\ell$.  We say that this cycle is the \emph{twin} of $C_i$ and
  denote it by $\twin(C_i) = (\twin(\nu_i^1), \dots,
  \twin(\nu_i^\ell))$ where $\twin(\nu_i^j)$ denotes the new vertex
  incident to the edge $e_j$.  As the interiors of $C_i$ and
  $\twin(C_i)$ are empty, we can glue the skeletons $\skel(\mu)$ and
  $\skel(\twin(\mu))$ together by identifying the vertices of $C_i$
  with the corresponding vertices in $\twin(C_i)$ (one of the
  mbeddings has to be flipped).  Applying this replacement for every
  virtual vertex and gluing it with its twin leads to an embedded
  planar graph $G^+$ with the following properties.  First, $G^+$
  contains a subdivision of $G$.  Second, for every cut corresponding
  to an edge $\eps = \{\mu, \mu_i\}$ in $T$, $G^+$ contains the cycle
  $C_i$ with exactly one subdivision vertex of an edge $e$ of $G$ if
  the cut corresponding to $\eps$ separates the endpoints of $e$.
  Third, no two of these cycles share a vertex.  The planar drawing of
  $G^+$ gives a planar drawing of $G$.  Moreover, the drawings of the
  cycles can be used as curves representing the cuts, yielding a
  c-planar drawing of $G$.
\end{proof}

\subsection{Cutvertices in Skeletons}

We show that cutvertices in skeletons correspond to different
connected components in a cluster or in a co-cluster.  More precisely,
a cutvertex directly implies disconnectivity, while the opposite is
not true.  Consider the example in
Fig.~\ref{fig:cluster-decomposition-tree}.  The cutvertex in the
skeleton containing the vertices~$7$ and~$8$ corresponds to the two
connected components in the blue cluster (containing $7$ and $8$).
However, the two connected components in the orange cluster
(containing $6$--$8$) do not yield a new cutvertex in the skeleton
containing the vertex~$6$.  The following lemma in particular shows
that requiring every cluster to be connected implies that the parent
vertices of skeletons cannot be cutvertices.

\begin{lemma}
  \label{lem:cutvertex-connected-components}
  Let $\nu$ be a virtual vertex that is a cutvertex in its skeleton.
  The expansion graphs of virtual vertices in different blocks
  incident to $\nu$ belong to different connected components in
  $\exp(\twin(\nu))$.
\end{lemma}
\begin{proof}
  Let $\mu$ be the node whose skeleton contains $\nu$.  Recall that
  one can obtain the graph $\exp(\twin(\nu))$ by removing $\nu$ from
  $\skel(\mu)$ and replacing all other virtual vertices of
  $\skel(\mu)$ with their expansion graphs.  Clearly, this yields (at
  least) one different connected component for each of the blocks
  incident to $\nu$.
\end{proof}

While the converse of Lemma~\ref{lem:cutvertex-connected-components}
is generally not true, it holds if the condition is satisfied for all
parent vertices in all skeletons simultaneously.

\begin{lemma}
  \label{lem:connected-clusters-no-cutvertices}
  Every cluster in a clustered graph is connected if and only if in
  every node $\mu$ of the rooted cd-tree the parent vertex is not a
  cutvertex in $\skel(\mu)$.
\end{lemma}
\begin{proof}
  By Lemma~\ref{lem:cutvertex-connected-components}, the existence of
  a cutvertex implies a disconnected cluster.  Conversely, let
  $\pert(\mu)$ be disconnected and assume without loss of generality
  that $\pert(\mu_i)$ is connected for every child $\mu_1, \dots,
  \mu_k$ of $\mu$ in the cd-tree.  One obtains $\skel(\mu)$ without
  the parent vertex $\nu$ by contracting in $\pert(\mu)$ the child
  clusters $\pert(\mu_i)$ to virtual vertices $\nu_i$.  As the
  contracted graphs $\pert(\mu_i)$ are connected while the initial
  graph $\pert(\mu)$ is not, the resulting graph must be disconnected.
  Thus, $\nu$ is a cutvertex in $\skel(\mu)$.
\end{proof}

\section{Clustered and Constrained Planarity}
\label{sec:clust-constr-plan}

We first describe several constraints on planar embeddings, each
restricting the edge-orderings of vertices.  We then show the relation
to c-planarity.

Consider a finite set $S$ (e.g., edges incident to a vertex).  Denote
the set of all cyclic orders of $S$ by $O_S$.  An
\emph{order-constraint} on $S$ is simply a subset of $O_S$ (only the
orders in the subset are \emph{allowed}).  A \emph{family of
  order-constraints} for the set $S$ is a set of different order
constraints, i.e., a subset of the power set of $O_s$.
We say that a family of order-constraints has a \emph{compact
  representation}, if one can specify every order-constraint in this
family with polynomial space (in $|S|$).  In the following we describe
families of order-constraints with compact representations.

A \emph{partition-constraint} is given by partitioning $S$ into
subsets $S_1 \cupdot \dots \cupdot S_k = S$.  It requires that no two
partitions \emph{alternate}, i.e., elements $a_i, b_i \in S_i$ and
$a_j, b_j \in S_j$ must not appear in the order $a_i, a_j, b_i, b_j$.
A \emph{PQ-constraint} requires that the order of elements in $S$ is
represented by a given PQ-tree with leaves~$S$.  A
\emph{full-constraint} contains only one order, i.e., the order of $S$
is completely fixed.

A \emph{partitioned full-constraint} restricts the orders of elements
in $S$ according to a partition constraint (partitions must not
alternate) and additionally completely fixes the order within each
partition.  Similarly, \emph{partitioned PQ-constraints} require the
elements in each partition to be ordered according to a PQ-constraint.
Clearly, this notion of partitioned order-constraints generalizes to
arbitrary order-constraints.


Consider a planar graph $G$.  By \emph{constraining} a vertex $v$ of
$G$, we mean that there is an order-constraint on the edges incident
to $v$.  We then only allow planar embeddings of $G$ where the
edge-ordering of $v$ is allowed by the order-constraint.  By
constraining $G$, we mean that several (or all) vertices of $G$ are
constrained.

\subsection{Flat-Clustered Graph}
\label{sec:flat-clustered-graph}

Consider a flat-clustered graph, i.e., a clustered graph where the
cd-tree is a star.  We choose the center $\mu$ of the star to be the
root.  Let $\nu_1, \dots, \nu_k$ be the virtual vertices in
$\skel(\mu)$ corresponding to the children $\mu_1, \dots, \mu_k$ of
$\mu$.  Note that $\skel(\mu_i)$ contains exactly one virtual vertex,
namely $\twin(\nu_i)$.  The possible ways to embed $\skel(\mu_i)$
restrict the possible edge-orderings of $\twin(\nu_i)$ and thus, by
the characterization in Theorem~\ref{thm:characterization}, the
edge-orderings of $\nu_i$ in $\skel(\mu)$.  Hence, the graph
$\skel(\mu_i)$ essentially yields an order constraint for $\nu_i$ in
$\skel(\mu)$.  We consider c-planarity with differently restricted
instances, leading to different families of order-constraints.  To
show that testing c-planarity is equivalent to testing whether
$\skel(\mu)$ is planar with respect to order-constraints of a specific
family, we have to show two directions.  First, the embeddings of
$\skel(\mu_i)$ only yield order-constraints of the given family.
Second, we can get every possible order-constraint of the given family
by choosing an appropriate graph for
$\skel(\mu_i)$.  

\begin{figure}[tb]
  \centering
  \includegraphics[page=1]{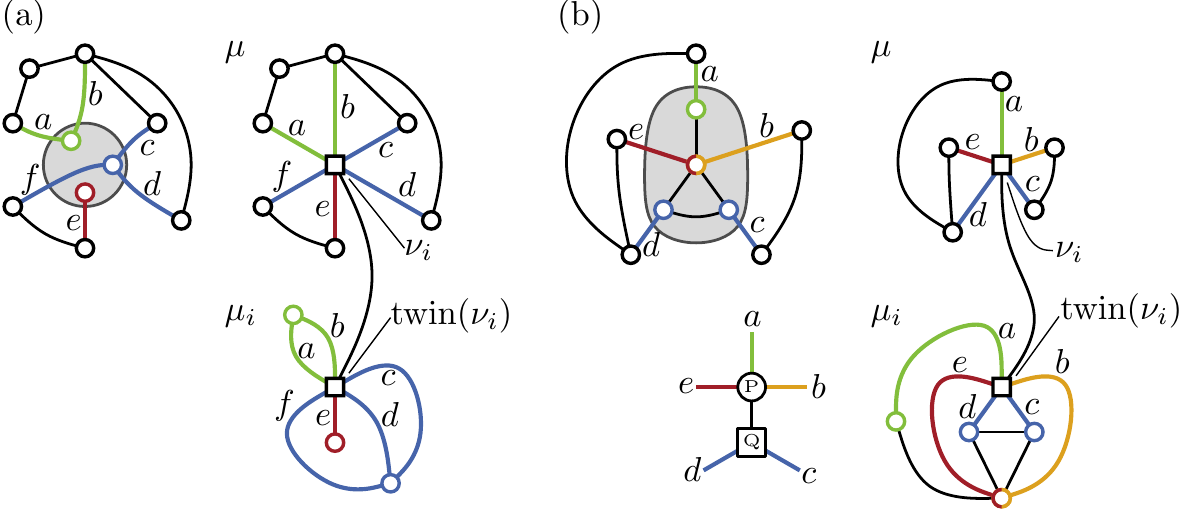}
  \caption{(a)~A graph with a single cluster consisting of isolated
    vertices together with an illustration of its cd-tree.  An
    edge-ordering of $\twin(\nu_i)$ corresponds to a planar embedding
    of $\skel(\mu_i)$ if and only if no two partitions of the
    partitioning $\{\{a, b\}, \{c, d, f\}, \{e\}\}$ alternate. (b)~A
    graph with a single connected cluster and its cd-tree.  The valid
    edge-orderings of $\twin(\nu_i)$ are represented by the shown
    PQ-tree.}
  \label{fig:flat-cluster-constrained-embedding}
\end{figure}

\begin{theorem}
  \label{thm:constrained-planarity-equivalence}
  Testing c-planarity of flat-clustered graphs
  \begin{inparaenum}[\normalfont(i)]
  \item \emph{where each proper cluster consists of isolated vertices;}
  \item \emph{where each cluster is connected;}
  \item \emph{with fixed planar embedding;}
  \item \emph{without restriction}
  \end{inparaenum}
  is linear-time equivalent to testing planarity of a multi-graph with
  \begin{inparaenum}[\normalfont(i)]
  \item \emph{partition-constraints;}
  \item \emph{PQ-constraints;}
  \item \emph{partitioned full-constraints;}
  \item \emph{partitioned PQ-constraints,}
  \end{inparaenum}
  respectively.
\end{theorem}
\begin{proof}
  We start with case (i); see
  Fig.~\ref{fig:flat-cluster-constrained-embedding}a.  Consider a
  flat-clustered graph $G$ and let $\mu_i$ be one of the leaves of the
  cd-tree.  As $\pert(\mu_i)$ is a proper cluster, it consists of
  isolated vertices.  Thus, $\skel(\mu_i)$ is a set of vertices $v_1,
  \dots, v_\ell$, each connected (with multiple edges) to the virtual
  vertex $\twin(\nu_i)$.  The vertices $v_1, \dots, v_\ell$ partition
  the edges incident to $\twin(\nu_i)$ into $\ell$ subsets.  Clearly,
  in every planar embedding of $\skel(\mu_i)$ no two partitions
  alternate.  Moreover, every edge-ordering of $\twin(\nu_i)$ in which
  no two partitions alternate gives a planar embedding of
  $\skel(\mu_i)$.  Thus, the edges incident to $\nu_i$ in $\skel(\mu)$
  are constrained by a partition-constraint, where the partitions are
  determined by the incidence of the edges to the vertices $v_1,
  \dots, v_\ell$.  One can easily construct the resulting instance of
  planarity with partition-constraints problem in linear time in the
  size of the cd-tree.  Note that the cd-tree has linear size in $G$
  for flat-clustered graphs.

  Conversely, given a planar graph $H$ with partition-constraints, we
  set $\skel(\mu) = H$.  For every vertex of $H$ we have a virtual
  vertex $\nu_i$ in $\skel(\mu)$ with corresponding child~$\mu_i$.  We
  can simulate every partitioning of the edges incident to $\nu_i$ by
  connecting edges incident to $\twin(\nu_i)$ (in $\skel(\mu_i)$) with
  vertices such that two edges are connected to the same vertex if and
  only if they belong to the same partition.  Clearly, this
  construction runs in linear time.

  Case (ii) is illustrated in
  Fig.~\ref{fig:flat-cluster-constrained-embedding}b.  By
  Lemma~\ref{lem:connected-clusters-no-cutvertices} the condition of
  connected clusters is equivalent to requiring that the virtual
  vertex $\twin(\nu_i)$ in the skeleton of any leaf $\mu_i$ of the
  cd-tree is not a cutvertex.  The statement of the theorem follows
  from the fact that the possible edge-orderings of non-cutvertices is
  PQ-representable and that any PQ-tree can be achieved by choosing an
  appropriate planar graph in which $\twin(\nu_i)$ is not a cutvertex
  (see Section~\ref{sec:preliminaries}).

  In case (iii) the embedding of $G$ is fixed.  As in case (i), the
  blocks incident to $\twin(\nu_i)$ in $\skel(\mu_i)$ partition the
  edges incident to $\nu_i$ in $\skel(\mu)$ such that two partitions
  must not alternate.  Moreover, the fixed embedding of $G$ fixes the
  edge-ordering of non-virtual vertices and thus it fixes the
  embeddings of the blocks in $\skel(\mu_i)$.  Hence, we get
  partitioned full-constraints for~$\nu_i$.  Conversely, we can
  construct an arbitrary partitioned full-constraint as in case~(i).

  For case (iv) the arguments from case (iii) show that we again get
  partitioned order-constraints, while the arguments from case (ii)
  show that these order-constraints (for the blocks) are
  PQ-constraints.
\end{proof}

\subsubsection{Related Work}
\label{sec:related-work-flat}

Biedl~\cite{b-dppi-98} proposes different drawing-models for graphs
whose vertices are partitioned into two subsets.  The model matching
the requirements of c-planar drawings is called \emph{HH-drawings}.
Biedl et al.~\cite{bkm-dppii-98} show that one can test for the
existence of HH-drawings in linear time.  Hong and
Nagamochi~\cite{hn-sattb-14} rediscovered this result in the context
of 2-page book embeddings.  These results solve c-planarity for
flat-clustered graphs if the skeleton of the root node contains only
two virtual vertices.  This is equivalent to testing planarity with
partitioned PQ-constraints for multi-graphs with only two vertices
(Theorem~\ref{thm:constrained-planarity-equivalence}).  Thus, to solve
c-planarity for flat-clustered graphs, one needs to solve an embedding
problem on general planar multi-graphs that is so far only solved on a
set of parallel edges (with absolutely non-trivial algorithms).  This
indicates that we are still far away from solving the c-planarity
problem even for flat-clustered graphs.

Cortese et al.~\cite{cbpp-cccc-05} give a linear-time algorithm for
testing c-planarity of a flat-clustered cycle (i.e., $G$ is a simple
cycle) if the skeleton of the cd-tree's root is a multi-cycle.  The
requirement that $G$ is a cycle implies that the skeleton of each
non-root node in $T$ has the property that the blocks incident to the
parent vertex are simple cycles.  Thus, in terms of constrained
planarity, they show how to test planarity of multi-cycles with
partition-constraints where each partition has size two.  The result
can be extended to a special case of c-planarity where the clustering
is not flat.  However, the cd-tree fails to have easy-to-state
properties in this case, which shows that the cd-tree perspective of
course has some limitations.
Later, Cortese et al.~\cite{cbpp-ecpg-09} extended this result to the
case where $G$ is still a cycle, while the skeleton of the root can be
an arbitrary planar multi-graph that has a fixed embedding up to the
ordering of parallel edges.  This is equivalent to testing planarity
of such a graph with partition-constraints where each partition has
size two.

Jel\'inkov\'a et al.~\cite{jkk-cp-09} consider the case where each
cluster contains at most three vertices (with additional
restrictions).  Consider a cluster containing only two vertices $u$
and $v$.  If $u$ and $v$ are connected, then the region representing
the cluster can always be added, and we can omit the cluster.
Otherwise, the region representing the cluster in a c-planar drawing
implies that one can add the edge $uv$ to $G$, yielding an equivalent
instance.  Thus, one can assume that every cluster has size
exactly~$3$, which yields flat-clustered graphs.  In this setting they
give efficient algorithms for the cases that $G$ is a cycle and that
$G$ is 3-connected.  Moreover, they give an FPT-algorithm for the case
that $G$ is an \emph{Eulerian graph} with $k$ nodes, i.e., a graph
obtained from a 3-connected graph of size $k$ by multiplying and then
subdividing edges.

In case $G$ is 3-connected, its planar embedding is fixed and thus the
edge-ordering of non-virtual vertices is fixed.  Thus, one obtains
partitioned full-constraints with the restriction that there are only
three partitions.  Clearly, the requirement that $G$ is 3-connected
also restricts the possible skeletons of the root of the cd-tree.  It
is an interesting open question whether planarity with partitioned
full-constraints with at most three partitions can be tested
efficiently for arbitrary planar graphs.  In case $G$ is a cycle, one
obtains partition constraints with only three partitions and each
partition has size two.  Note that this in particular restricts the
skeleton of the root to have maximum degree~6.  Although these kind of
constraints seem pretty simple to handle, the algorithm by
Jel\'inkov\'a et al. is quite involved.  It seems like one barrier
where constrained embedding becomes difficult is when there are
partition constraints with three or more partitions (see also
Theorem~\ref{thm:general-algorithm-sim-pq-ord}).  The result about
Eulerian graphs in a sense combines the cases where $G$ is 3-connected
and a cycle.  A vertex has either degree two and thus yields a
partition of size two or it is one of the constantly many vertices
with higher degree, for which the edge-ordering is partly fixed.

Chimani et al.~\cite{cdfk-atcpe-14} give a polynomial-time algorithm
for embedded flat-clustered graphs with the additional requirement
that each face is incident to at most two vertices of the same
cluster.  This basically solves planarity with partitioned
full-constraints with some additional requirements.  We do not
describe how these additional requirements exactly restrict the
possible instances of constrained planarity.  However, we give some
properties that shed a light on why these requirements make planarity
with partitioned full-constraints tractable.

Consider the skeleton $\skel(\mu)$ of a (non-root) node $\mu$ of the
cd-tree.  As $G$ is a flat-clustered graph, $\skel(\mu)$ has only a
single virtual vertex.  Assume we choose planar embeddings with
consistent edge orderings for all skeletons (i.e., we have a c-planar
embedding of $G$).  Two non-virtual vertices $u$ and $v$ in
$\skel(\mu)$ that are incident to the same face of $\skel(\mu)$ are
then also incident to the same face of $G$.  Note that the converse is
not true, as two vertices that share a face in $G$ may be separated in
$\skel(\mu)$ due to the contraction of disconnected subgraphs.  As the
non-virtual vertices of $\skel(\mu)$ belong to the same cluster, at
most two of them can be incident to a common face of $\skel(\mu)$.
Thus, every face of $\skel(\mu)$ has at most two non-virtual vertices
on its boundary.  One implication of this fact is that every connected
component of the cluster is a tree for the following reason.  If a
connected component contains a cycle, it has at least two faces with
more than two vertices on the boundary.  In $\skel(\mu)$ only one of
the two faces can be split into several faces by the virtual vertex,
but the other face remains.

More importantly, the possible ways how the blocks incident to the
virtual vertex of $\skel(\mu)$ can be nested into each other is
heavily restricted.  In particular, embedding multiple blocks next to
each other into the same face of another block is not possible, as
this would result in a face of $\skel(\mu)$ with more than two
non-virtual vertices on its boundary.  In a sense, this enforces a
strong nesting of the blocks.  Thus, one actually obtains a variant of
planarity with partitioned full-constraints, where the way how the
partitions can nest is restricted beyond forbidding two partitions to
alternate.  These and similar restrictions on how partitions are
allowed to be nested lead to a variety of new constrained planarity
problems.  We believe that studying these restricted problems can help
to deepen the understanding of the more general partitioned
full-constraints or even partitioned PQ-constraints.

\subsection{General Clustered Graphs}
\label{sec:gener-clust-graphs}

Expressing c-planarity for general clustered graphs (not necessarily
flat) in terms of constrained planarity problems is harder for the
following reason.  Consider a leaf $\mu$ in the cd-tree.  The skeleton
of $\mu$ is a planar graph yielding (as in the flat-clustered case)
partitioned PQ-constraints for its parent $\mu'$.  This restricts the
possible embeddings of $\skel(\mu')$ and thus the order-constraints
one obtains for the parent of $\mu'$ are not necessarily again
partitioned PQ-constraints.

One can express this issue in the following, more formal way.  Let
$G$ be a planar multi-graph with vertices $v_1, \dots, v_n$ and
designated vertex $v = v_n$.  The map $\varphi_G^v$ maps a tuple
$(C_1, \dots, C_n)$ where $C_i$ is an order-constraint on the edges
incident to $v_i$ to an order-constraint $C$ on the edges incident to
$v$.  The order-constraint $C = \varphi_G^v(C_1, \dots, C_n)$ contains
exactly those edge-orderings of $v$ that one can get in a planar
embedding of $G$ that respects $C_1, \dots, C_n$.  Note that $C$ is
empty if and only if there is no such embedding.  Note further that
testing planarity with order-constraints is equivalent to deciding
whether $\varphi_G^v$ evaluates to the empty set.  We call such a map
$\varphi_G^v$ a \emph{constrained-embedding operation}.

The issue mentioned above (with constraints iteratively handed to the
parents) boils down to the fact that partitioned PQ-constraints are
not closed under constrained-embedding operations.  On the positive
side, we obtain a general algorithm for solving c-planarity as
follows.  Assume we have a family of order-constraints $\mathcal C$
with compact representations that is closed under
constrained-embedding operations.  Assume further that we can evaluate
the constrained embedding operations in polynomial time on
order-constraints in $\mathcal C$.  Then one can simply solve
c-planarity by traversing the cd-tree bottom-up, evaluating for a node
$\mu$ with parent vertex $\nu$ the constrained-embedding operation
$\varphi_{\skel(\mu)}^\nu$ on the constraints one computed in the same
way for the children of $\mu$.

Clearly, when restricting the skeletons of the cd-tree or requiring
properties for the parent vertices in these skeletons, these
restrictions carry over to the constrained-embedding operations one
has to consider.  More precisely, let $\mathcal R$ be a set of pairs
$(G, v)$, where $v$ is a vertex in $G$.  We say that a clustered graph
is \emph{$\mathcal R$-restricted} if $(\skel(\mu), \nu) \in \mathcal
R$ holds for every node $\mu$ in the cd-tree with parent vertex $\nu$.
Moreover, the \emph{$\mathcal R$-restricted} constrained-embedding
operations are those operations $\varphi_G^v$ with $(G, v) \in
\mathcal R$.  The following theorem directly follows.

\begin{theorem}
  \label{thm:solving-c-planarity-bottom-up}
  One can solve c-planarity of $\mathcal R$-restricted clustered
  graphs in polynomial time if there is a family $\mathcal C$ of
  order-constraints such that
  \begin{itemize}
  \item $\mathcal C$ has a compact representation,
  \item $\mathcal C$ is closed under $\mathcal R$-restricted
    constrained-embedding operations,
  \item every $\mathcal R$-restricted constrained-embedding operation
    on order-constraints in $\mathcal C$ can be evaluated in
    polynomial time.
  \end{itemize}
\end{theorem}

When dropping the requirement that $\mathcal C$ has a compact
representation the algorithm becomes super-polynomial only in the
maximum degree $d$ of the virtual vertices (the number of possible
order-constraints for a set of size $d$ depends only on $d$).
Moreover, if the input of $\varphi_G^v$ consists of only $k$ order
constraints (whose sizes are bounded by a function of $d$), then
$\varphi_G^v$ can be evaluated by iterating over all combinations of
orders, applying a planarity test in every step.  This gives an
FPT-algorithm with parameter $d + k$ (running time $O(f(d + k)p(n))$,
where $f$ is a computable function depending only on $d + k$ and $p$
is a polynomial).  In other words, we obtain an FPT-algorithm where
the parameter is the sum of the maximum degree of the tree~$T$ and the
maximum number of edges leaving a cluster.  Note that this generalizes
the FPT-algorithm by Chimani and Klein~\cite{ck-ssscp-13} with respect
to the total number of edges connecting different clusters.

Moreover, Theorem~\ref{thm:solving-c-planarity-bottom-up} has the
following simple implication.  Consider a clustered graph where each
cluster is connected.  This restricts the skeletons of the cd-tree
such that non of the parent vertices is a cutvertex
(Lemma~\ref{lem:cutvertex-connected-components}).  Thus, we have
$\mathcal R$-restricted clustered graphs where $(G, v) \in \mathcal R$
implies that $v$ is not a cutvertex in $G$.  PQ-constraints are closed
under $\mathcal R$-restricted constrained-embedding operations as the
valid edge-orderings of non-cutvertices are PQ-representable and
planarity with PQ-constraints is basically equivalent to planarity
(one can model a PQ-tree with a simple gadget; see
Section~\ref{sec:preliminaries}).  Thus,
Theorem~\ref{thm:solving-c-planarity-bottom-up} directly implies that
c-planarity can be solved in polynomial time if each cluster is
connected.

\subsubsection{Related Work}
\label{sec:related-work-connected}

The above algorithm resulting from
Theorem~\ref{thm:solving-c-planarity-bottom-up} is more or less the
one described by Lengauer~\cite{l-hpta-89}.  The algorithm was later
rediscovered by Feng et al.~\cite{fce-pcg-95} who coined the term
``c-planarity''.  The algorithm runs in $O(c) \subseteq O(n^2)$ time
(recall that $c$ is the size of the cd-tree).
Dahlhaus~\cite{d-ltarc-98} improves the running time to $O(n)$.
Cortese et al.\cite{cbf-cpcccg-08} give a characterization that also
leads to a linear-time algorithm.

Goodrich et al.~\cite{gls-cpecg-06} consider the case where each
cluster is either connected or \emph{extrovert}.  Let $\mu$ be a node
in the cd-tree with parent $\mu'$.  The cluster $\pert(\mu)$ is
extrovert if the parent cluster $\pert(\mu')$ is connected and every
connected component in $\pert(\mu)$ is connected to a vertex not in
the parent $\pert(\mu')$.  They show that one obtains an equivalent
instance by replacing the extrovert cluster $\pert(\mu)$ by one
cluster for each of its connected components while requiring
additional PQ-constraints for the parent vertex in the resulting
skeleton.  In this instance every cluster is connected and the
additional PQ-constraints clearly do no harm.

Another extension to the case where every cluster must be connected is
given by Gutwenger et al.~\cite{gjl-acptcg-02}.  They give an
algorithm for the case where every cluster is connected with the
following exception.  Either, the disconnected clusters form a path in
the tree or for every disconnected cluster the parent and all siblings
are connected.  This has basically the effect that at most one
order-constraint in the input of a constrained-embedding operation is
not a PQ-tree.

Jel\'inek et al.~\cite{jjkl-cp-09-long,jjkl-cp-09} assume each cluster
to have at most two connected components and the underlying
(connected) graph to have a fixed planar embedding.  Thus, they
consider $\mathcal R$-restricted clustered graphs where $(G, v) \in
\mathcal R$ implies that $v$ is incident to at most two different
blocks.  The fixed embedding of the graph yields additional
restrictions that are not so easy to state within this model.

\section{Cutvertices with Two Non-Trivial Blocks}
\label{sec:c-planarity-pq}

The input of the \textsc{Simultaneous PQ-Ordering} problem consists of
several PQ-trees together with child-parent relations between them
(the PQ-trees are the nodes of a directed acyclic graph) such that the
leaves of every child form a subset of the leaves of its parents.
\textsc{Simultaneous PQ-Ordering} asks whether one can choose orders
for all PQ-trees \emph{simultaneously} in the sense that every
child-parent relation implies that the order of the leaves of the
parent are an extension of the order of the leaves of the child.  In
this way one can represent orders that cannot be represented by a
single PQ-tree.  For example, adding one or more children to a PQ-tree
$T$ restricts the set of orders represented by $T$ by requiring the
orders of different subsets of leaves to be represented by some other
PQ-tree.  Moreover, one can synchronize the orders of different trees
that share a subset of leaves by introducing a common child containing
these leaves.

\textsc{Simultaneous PQ-Ordering} is NP-hard but efficiently solvable
for so-called 2-fixed instances~\cite{br-spoacep-13}.  For every
biconnected planar graph $G$, there exists an instance of
\textsc{Simultaneous PQ-Ordering}, the \emph{PQ-embedding
  representation}, that represents all planar embeddings of
$G$~\cite{br-spoacep-13}.  It has the following properties.
\begin{itemize}
\item For every vertex $v$ in $G$ there is a PQ-tree $T(v)$, the
  \emph{embedding tree}, that has the edges incident to $v$ as leaves.
\item 
  For every solution of the PQ-embedding representation, setting the
  edge-ordering of every vertex $v$ to the order given by $T(v)$
  yields a planar embedding.  Moreover, one can obtain every embedding
  of $G$ in this way.
\item The instance remains 2-fixed 
  when adding up to one child to each embedding tree.
\end{itemize}
A PQ-embedding representation still exists if every cutvertex in $G$
is incident to at most two \emph{non-trivial blocks} (blocks that are
not just bridges)\cite{br-spoacep-11}.

\begin{theorem}
\label{thm:general-algorithm-sim-pq-ord}
C-planarity can be tested in $O(c^2) \subseteq O(n^4)$ time if every
virtual vertex in the skeletons of the cd-tree is incident to at most
two non-trivial blocks.
\end{theorem}
\begin{proof}
  Let $G$ be a clustered graph with cd-tree $T$.  For the skeleton of
  each node in~$T$, we get a PQ-embedding representation with the
  above-mentioned properties.  Let $\mu$ be a node of $T$ and let
  $\nu$ be a virtual vertex in $\skel(\mu)$.  By the above properties,
  the embedding representation of $\mu$ contains the embedding tree
  $T(\nu)$ representing the valid edge-orderings of $\nu$.  Moreover,
  for $\twin(\nu)$ there is an embedding tree $T(\twin(\nu))$ in the
  embedding representation of the skeleton containing $\twin(\nu)$.
  To ensure that $\nu$ and $\twin(\nu)$ have the same edge-ordering,
  one can simply add a PQ-tree as a common child of $T(\nu)$ and
  $T(\twin(\nu))$.  We do this for every virtual node in the skeletons
  of $T$.  Due to the last property of the PQ-embedding
  representations, the resulting instance remains 2-fixed and can thus
  be solved efficiently.

  Every solution of this \textsc{Simultaneous PQ-Ordering} instance
  $D$ yields planar embeddings of the skeletons such that every
  virtual vertex and its twin have the same edge-ordering.
  Conversely, every such set of embeddings yields a solution for $D$.
  It thus follows by the characterization in
  Theorem~\ref{thm:characterization} that solving $c$-planarity is
  equivalent to solving $D$.  The size of $D$ is linear in the size
  $c$ of the cd-tree $T$.  Moreover, solving \textsc{Simultaneous
    PQ-Ordering} for 2-fixed instances can be done in quadratic
  time~\cite{br-spoacep-13}, yielding the running time $O(c^2)$.
\end{proof}

Theorem~\ref{thm:general-algorithm-sim-pq-ord} includes the following
interesting cases.  The latter extends the result by Jel\'inek et
al.~\cite{jstv-cp-09} from four to five outgoing edges per cluster.
  
\begin{corollary}
  C-planarity can be tested in $O(c^2) \subseteq O(n^4)$ time if every
  cluster and every co-cluster has at most two connected components.
\end{corollary}
\begin{proof}
  Note that the expansion graphs of nodes in skeletons of $T$ are
  exactly the clusters and co-clusters.  Thus, the expansion graphs
  consist of at most two connected components.  By
  Lemma~\ref{lem:cutvertex-connected-components} the cutvertices in
  skeletons of $T$ are incident to at most two different blocks.
  Thus, we can simply apply
  Theorem~\ref{thm:general-algorithm-sim-pq-ord}.
\end{proof}

\begin{corollary}
  \label{cor:degree-5}
  C-planarity can be tested in $O(n^2)$ time if every cluster has at
  most five outgoing edges.
\end{corollary}
\begin{proof}
  Let $\mu$ be a node with virtual vertex $\nu$ in its skeleton.  The
  edges incident to $\nu$ in $\skel(\mu)$ are exactly the edges that
  separate $\exp(\nu)$ from the rest of the graph $\exp(\twin(\nu))$.
  Thus, if every cluster has at most five outgoing edges, the virtual
  vertices in skeletons of $T$ have maximum degree~5.  With five edges
  incident to a vertex $\nu$, one cannot get more than two non-trivial
  blocks incident to $\nu$.  It follows from
  Theorem~\ref{thm:general-algorithm-sim-pq-ord} that we can test
  c-planarity in $O(c^2)$ time.  As we have a linear number of cuts,
  each of constant size (at most~$5$), we get $c \in O(n)$.
\end{proof}

\section{Conclusion}
\label{sec:conclusion}

In this paper we have introduced the cd-tree and we have shown that it
can be used to reformulate the classic c-planarity problem as a
constrained embedding problem.  Afterwards, we interpreted several
previous results on c-planarity from this new perspective.  In in many
cases the new perspective simplifies these algorithms or at least
gives a better intuition why the imposed restrictions are helpful
towards making the problem tractable.  In some cases, the new view
allowed us to generalize and extend previous results to larger sets of
instances.

We believe that the constrained embedding problems we defined provide
a promising starting point for further research, e.g., by studying
restricted variants to further deepen the understanding of the
c-planarity problem. 

\ifdefined\LNCS
\bibliographystyle{splncs03}
\fi
\ifdefined\arXiv
\bibliographystyle{plain}
\fi
\bibliography{strings,c-planarity-pq}

\newcommand{\bibsoda}[2]{Proceedings of the #1 Annual ACM-SIAM Symposium on
  Discrete Algorithms (SODA'#2)} \newcommand{\bibgd}[2]{Proceedings of the #1
  International Symposium on Graph Drawing (GD'#2)}
  \newcommand{\bibinfovis}[1]{Proceedings of the IEEE Symposium on Information
  Visualization (InfoVis'#1)} \newcommand{\bibvis}[1]{Proceedings of the IEEE
  Conference on Visualization (Vis'#1)} \newcommand{\bibpvis}[1]{Proceedings of
  the IEEE Pacific Visualisation Symposium (PacificVis'#1)}
  \newcommand{\bibsoftvis}[2]{Proceedings of the #1 ACM Symposium on Software
  Visualization (SoftVis'#2)} \newcommand{\bibeurocg}[2]{Proceedings of the #1
  European Workshop on Computational Geometry (EuroCG'#2)}
  \newcommand{\bibsocg}[2]{Proceedings of the #1 Annual Symposium on
  Computational Geometry (SoCG'#2)} \newcommand{\bibwads}[2]{Proceedings of the
  #1 International Symposium on Algorithms and Data Structures (WADS'#2)}
  \newcommand{\bibwg}[2]{Proceedings of the #1 Workshop on Graph-Theoretic
  Concepts in Computer Science (WG'#2)} \newcommand{\bibgta}{Proceedings of the
  Conference at Graph Theory and Applications}
  \newcommand{\bibisaac}[2]{Proceedings of the #1 International Symposium on
  Algorithms and Computation (ISAAC'#2)} \newcommand{\bibcocoon}[2]{Proceedings
  of the #1 Annual International Conference on Computing and Combinatorics
  (COCOON'#2)} \newcommand{\bibtamc}[2]{Proceedings of the #1 Annual Conference
  on Theory and Applications of Models of Computation (TAMC'#2)}
  \newcommand{\bibicalp}[2]{Proceedings of the #1 International Colloquium on
  Automata, Languages and Programming (ICALP'#2)}
  \newcommand{\biblatin}[2]{Proceedings of the #1 Latin American Symposium
  (LATIN'#2)} \newcommand{\bibesa}[2]{Proceedings of the #1 Annual European
  Symposium on Algorithms (ESA'#2)}
\begin{thebibliography}{10}

\bibitem{b-dppi-98}
Therese Biedl.
\newblock Drawing planar partitions {I}: {LL}-drawings and {LH}-drawings.
\newblock In {\em \bibsocg{14th}{98}}, pages 287--296. ACM Press, 1998.

\bibitem{bkm-dppii-98}
Therese Biedl, Michael Kaufmann, and Petra Mutzel.
\newblock Drawing planar partitions {II}: {HH}-drawings.
\newblock In Juraj Hromkovi{\v{c}} and Ondrej S{\'y}kora, editors, {\em
  \bibwg{24th}{98}}, volume 1517 of {\em Lecture Notes in Computer Science},
  pages 124--136. Springer Berlin/Heidelberg, 1998.

\bibitem{br-spoacep-11}
Thomas Bl{\"a}sius and Ignaz Rutter.
\newblock Simultaneous {PQ}-ordering with applications to constrained embedding
  problems.
\newblock {\em Computing Research Repository}, abs/1112.0245:1--46, 2011.

\bibitem{br-spoacep-13}
Thomas Bl{\"a}sius and Ignaz Rutter.
\newblock Simultaneous {PQ}-ordering with applications to constrained embedding
  problems.
\newblock In {\em \bibsoda{24th}{13}}. Society for Industrial and Applied
  Mathematics, 2013.

\bibitem{bl-tcopi-76}
Kellogg~S. Booth and George~S. Lueker.
\newblock Testing for the consecutive ones property, interval graphs, and graph
  planarity using {PQ}-tree algorithms.
\newblock {\em Journal of Computer and System Sciences}, 13(3):335--379, 1976.

\bibitem{cdfk-atcpe-14}
Markus Chimani, Giuseppe Di~Battista, Fabrizio Frati, and Karsten Klein.
\newblock Advances on testing c-planarity of embedded flat clustered graphs.
\newblock In Christian Duncan and Antonios Symvonis, editors, {\em
  \bibgd{22nd}{14}}, volume 8871 of {\em Lecture Notes in Computer Science},
  pages 416--427. Springer Berlin/Heidelberg, 2014.

\bibitem{ck-ssscp-13}
Markus Chimani and Karsten Klein.
\newblock Shrinking the search space for clustered planarity.
\newblock In Walter Didimo and Maurizio Patrignani, editors, {\em
  \bibgd{20th}{12}}, volume 7704 of {\em Lecture Notes in Computer Science},
  pages 90--101. Springer Berlin/Heidelberg, 2013.

\bibitem{cw-cccg-06}
Sabine Cornelsen and Dorothea Wagner.
\newblock Completely connected clustered graphs.
\newblock {\em Journal of Discrete Algorithms}, 4(2):313--323, 2006.

\bibitem{cbf-cpcccg-08}
Pier~Francesco Cortese, Giuseppe {Di Battista}, Fabrizio Frati, Maurizio
  Patrignani, and Maurizio Pizzonia.
\newblock C-planarity of c-connected clustered graphs.
\newblock {\em Journal of Graph Algorithms and Applications}, 12(2):225--262,
  2008.

\bibitem{cbpp-cccc-05}
Pier~Francesco Cortese, Giuseppe {Di Battista}, Maurizio Patrignani, and
  Maurizio Pizzonia.
\newblock Clustering cycles into cycles of clusters.
\newblock {\em Journal of Graph Algorithms and Applications}, 9(3):391--413,
  2005.

\bibitem{cbpp-ecpg-09}
Pier~Francesco Cortese, Giuseppe {Di Battista}, Maurizio Patrignani, and
  Maurizio Pizzonia.
\newblock On embedding a cycle in a plane graph.
\newblock {\em Discrete Mathematics}, 309(7):1856--1869, 2009.

\bibitem{d-ltarc-98}
Elias Dahlhaus.
\newblock A linear time algorithm to recognize clustered planar graphs and its
  parallelization.
\newblock In Cl{\'a}udio~L. Lucchesi and Arnaldo~V. Moura, editors, {\em
  \biblatin{3rd}{98}}, volume 1380 of {\em Lecture Notes in Computer Science},
  pages 239--248. Springer Berlin/Heidelberg, 1998.

\bibitem{dt-omtc-96}
G.~{Di Battista} and R.~Tamassia.
\newblock On-line maintenance of triconnected components with {SPQR}-trees.
\newblock {\em Algorithmica}, 15(4):302--318, 1996.

\bibitem{g-ecpte-08}
Giuseppe {Di Battista} and Fabrizio Frati.
\newblock Efficient c-planarity testing for embedded flat clustered graphs with
  small faces.
\newblock In Seok-Hee Hong, Takao Nishizeki, and Wu~Quan, editors, {\em
  \bibgd{15th}{07}}, volume 4875 of {\em Lecture Notes in Computer Science},
  pages 291--302. Springer Berlin/Heidelberg, 2008.

\bibitem{fce-pcg-95}
Qing-Wen Feng, Robert~F. Cohen, and Peter Eades.
\newblock Planarity for clustered graphs.
\newblock In Paul Spirakis, editor, {\em \bibesa{3rd}{95}}, volume 979 of {\em
  Lecture Notes in Computer Science}, pages 213--226. Springer
  Berlin/Heidelberg, 1995.

\bibitem{gls-cpecg-06}
Michael~T. Goodrich, George~S. Lueker, and Jonathan~Z. Sun.
\newblock C-planarity of extrovert clustered graphs.
\newblock In Patrick Healy and Nikola~S. Nikolov, editors, {\em
  \bibgd{13th}{05}}, volume 3843 of {\em Lecture Notes in Computer Science},
  pages 211--222. Springer Berlin/Heidelberg, 2006.

\bibitem{gjl-acptcg-02}
Carsten Gutwenger, Michael J{\"u}nger, Sebastian Leipert, Petra Mutzel, Merijam
  Percan, and Ren{\'e} Weiskircher.
\newblock Advances in c-planarity testing of clustered graphs.
\newblock In Michael~T. Goodrich and Stephen~G. Kobourov, editors, {\em
  \bibgd{10th}{02}}, volume 2528 of {\em Lecture Notes in Computer Science},
  pages 220--235. Springer Berlin/Heidelberg, 2002.

\bibitem{hn-sattb-14}
Seok-Hee Hong and Hiroshi Nagamochi.
\newblock Simpler algorithms for testing two-page book embedding of partitioned
  graphs.
\newblock In Zhipeng Cai, Alex Zelikovsky, and Anu Bourgeois, editors, {\em
  Proceedings of the 20th International Symposium on Computing and
  Combinatorics (COCOON'14)}, volume 8591 of {\em Lecture Notes in Computer
  Science}, pages 477--488. Springer, 2014.

\bibitem{jjkl-cp-09-long}
V{\'i}t Jel{\'i}nek, Eva Jel{\'i}nkov{\'a}, Jan Kratochv{\'i}l, and Bernard
  Lidick{\'y}.
\newblock Clustered planarity: Embedded clustered graphs with two-component
  clusters.
\newblock Manuscript, 2009.

\bibitem{jjkl-cp-09}
V{\'i}t Jel{\'i}nek, Eva Jel{\'i}nkov{\'a}, Jan Kratochv{\'i}l, and Bernard
  Lidick{\'y}.
\newblock Clustered planarity: Embedded clustered graphs with two-component
  clusters (extended abstract).
\newblock In Ioannis~G. Tollis and Maurizio Patrignani, editors, {\em
  \bibgd{16th}{08}}, volume 5417 of {\em Lecture Notes in Computer Science},
  pages 121--132. Springer Berlin/Heidelberg, 2009.

\bibitem{jstv-cp-09}
V{\'i}t Jel{\'i}nek, Ond{\v{r}}ej Such{\'y}, Marek Tesa{\v{r}}, and
  Tom{\'a}{\v{s}} Vysko{\v{c}}il.
\newblock Clustered planarity: Clusters with few outgoing edges.
\newblock In Ioannis~G. Tollis and Maurizio Patrignani, editors, {\em
  \bibgd{16th}{08}}, volume 5417 of {\em Lecture Notes in Computer Science},
  pages 102--113. Springer Berlin/Heidelberg, 2009.

\bibitem{jkk-cp-09}
Eva Jel{\'i}nkov{\'a}, Jan K{\'a}ra, Jan Kratochv{\'i}l, Martin Pergel,
  Ond{\v{r}}ej Such{\'y}, and Tom{\'a}{\v{s}} Vysko{\v{c}}il.
\newblock Clustered planarity: Small clusters in cycles and eulerian graphs.
\newblock {\em Journal of Graph Algorithms and Applications}, 13(3):379--422,
  2009.

\bibitem{l-hpta-89}
Thomas Lengauer.
\newblock Hierarchical planarity testing algorithms.
\newblock {\em Journal of the ACM}, 36(3):474--509, 1989.

\end{thebibliography}

\end{document}